\documentclass[10pt,conference]{IEEEtran}
\usepackage{amssymb,pb-diagram,graphicx,verbatim}
\usepackage[tight]{subfigure}

\newcommand{\probup}{p_{node}}
\newcommand{\objup}{p_{obj}}

\newcommand{\FF}{\mathbb{F}}

\newcommand{\ov}{\mathbf{o}}

\newcommand{\Pc}{\mathcal{P}}
\newcommand{\Sc}{\mathcal{S}}

\newtheorem{definition}{Definition}
\newtheorem{example}{Example}
\newtheorem{theorem}{Theorem}
\newtheorem{lemma}{Lemma}

\title{Self-Repairing Codes for Distributed Storage\\ --- A Projective Geometric Construction}
\author{
\IEEEauthorblockN{Fr\'ed\'erique Oggier}
\IEEEauthorblockA{Division of Mathematical Sciences\\
School of Physical and Mathematical Sciences \\
Nanyang Technological University, Singapore\\
Email:~frederique@ntu.edu.sg}
\and
\IEEEauthorblockN{Anwitaman Datta}
\IEEEauthorblockA{
Division of Computer Science\\
School of Computer Engineering \\
Nanyang Technological University, Singapore\\
Email:~anwitaman@ntu.edu.sg}
}
\begin{document}
\maketitle

\begin{abstract}
Self-Repairing Codes (SRC) are codes designed to suit the need of coding for distributed networked storage:
they not only allow stored data to be recovered even in the presence of node failures, they also provide a repair
mechanism where as little as two live nodes can be contacted to regenerate the data of a failed node. In this paper, we propose a new
instance of self-repairing codes, based on constructions of spreads coming from projective geometry. We study some of their properties to demonstrate the suitability of these codes for distributed networked storage.

\end{abstract}

%
%
\section{Introduction}\label{sec:intro}

Storing digital data is a basic necessity of modern societies. The volume of data to be stored is tremendous, and is rapidly increasing. The kinds of data vary widely - from corporate and financial data repositories, archive of electronic communications to personal pictures, videos and work documents stored and shared in Web 2.0 and cloud based services, and much more. Distribution of such huge amount of data over multiple networked storage devices is thus the only practical and scalable solution.

All across the wide gamut of networked distributed storage systems design space, eventual failure of any and all individual storage devices is a given. Consequently, storing data redundantly is essential for fault tolerance. Furthermore, over a period of time, due to failures or departure of storage devices from the system, the redundancy will gradually decrease - risking the loss of the stored data, unless the redundancy is recreated. The possible ways for recreating redundancy depends on, to start with, the kind of redundancy being used.

Data redundancy can be achieved using replication - however that entails a very large storage overhead. Erasure coding based strategies in contrast can provide very good amount of redundancy for a very low storage overhead. However, when an encoded data block is lost and needs to be recreated, for traditional erasure codes, one would first need data equivalent in amount to recreate the whole object in one place (either by storing a full copy of the data, or else by downloading adequate encoded blocks), even in order to recreate a single encoded block. Such drawback of traditional erasure codes has in recent years given rise to a new flavor of coding research: designing erasure codes which need much less information to carry out the recreation of a lost encoded block.

%
\subsection{Related Work}

More precisely, consider a network of $n$ storage nodes, each with storage capacity $\alpha$,
where an object $\ov$ of size $B$ has to be stored.
A source possibly processes (encodes) the object $\ov$, splits it into
$n$ blocks, each of size at most $\alpha$, and stores such blocks at $n$ storage nodes.
When a data collector wants to retrieve the object, he should be able
to do so by contacting a subset of live nodes. We define $k$ as the minimum number of nodes that need to be contacted to retrieve the object, where the data collector may download upto $k\alpha$ amount of data,
and possibly process (decode) the downloaded data. For maximal distance separable (MDS) erasure codes, any arbitrary subset of $k$ nodes allow data retrievability. New nodes joining the network are assumed to perform the repair by contacting $d$ live nodes, from
each of which they download $\beta$ amount of data.

There are arguably two extreme points possible in the design-space of codes for distributed networked storage:

\textbf{(i)} Minimize the absolute amount of data transfer $d\beta$ needed to recreate the lost data from one node. Network-coding inspired analysis determines the storage-bandwidth (per repair) trade-offs, and a new family of codes called \emph{regenerating codes} (RGC) \cite{DGWK-journal,RSVR-allerton09} have been proposed, which can achieve (some points on) such a trade-off curve, under the assumption that $d\geq k$. Regenerating codes, like MDS erasure codes, allow data retrievability from any arbitrary set of $k$ nodes.

\textbf{(ii)} Minimize the number of nodes to be contacted for repairing one node failure. Recently proposed \emph{self-repairing codes} \cite{SRC} achieve this optimal, by allowing one repair while contacting only two nodes, i.e. $d=2$.
More specifically, self-repairing codes satisfy two cardinal properties, namely: (a) repairs can be performed directly through other subsets of nodes, without having to download data equivalent to that needed to reconstruct first the original object, ensuring that (b) a block is repaired from a fixed number of blocks, the number depending only on how many blocks are missing and independent of which specific blocks are missing. Note that minimization of the number of contacted nodes for a repair is achieved when the fixed number in clause (b) is in fact two.

Homomorphic self-repairing codes (HSRC) were proposed in \cite{SRC}, which, besides satisfying the cardinal properties elaborated above, were shown to (i) have $(n-1)/2$ distinct pairs with which the data for a single missing node could be regenerated, and consequently, (ii) self-repair for up to $(n-1)/2$ node failures could be carried out simultaneously using two nodes for each self-repair, from the pool of the remaining $(n+1)/2$ live nodes.

\subsection{Contributions}

This paper proposes a new family of self-repairing codes (PSRC) derived from a projective geometric construction. Besides the fundamental difference in its construction, and apart from satisfying the cardinal properties of self-repairing codes, as well as, in fact the other properties satisfied by HSRC, PSRC has several other salient features, as summarized next:

(i) Both the encoding and self-repair processes for PSRC involve only XOR operations, unlike HSRC encoding which involved the relatively more expensive task of evaluating a polynomial.

(ii) Similar to regenerating codes \cite{DGWK-journal}, in PSRC, each encoded block (i.e., data stored by a node) comprise of several ($\alpha$) pieces. Regeneration of the whole encoded block thus can likewise be done by regenerating the individual constituent pieces. This is in contrast to HSRC, where the encoded blocks were `atomic', and hence repair of the whole encoded block had to be carried out atomically. This gives PSRC some of the advantages of regenerating codes, while also naturally retaining the advantages of self-repairing codes, and provides several additional desirable properties, as elaborated next.

(iii) For self-repair of a specific node, if one live node is chosen arbitrarily, then there are several other nodes with which the first chosen node can be paired to regenerate the lost encoded block. This is in contrast to HSRC, where there is a unique pairing for one lost node, once one live node is chosen.

(iv) While the resulting code is strictly speaking not systematic in terms of what is stored at each node, if the constituent pieces stored over the nodes are considered, then systematic reconstruction of the object is possible, though this will need communication with $\alpha k << n$ specific nodes.

%
%
\section{Background from Projective Geometry}
\label{sec:PG}

The proposed construction as described in next section relies on the notion of
spread coming from projective geometry. We thus start by providing the required background.

Consider the finite field $\FF_q$, where
$q$ is a power of a prime $p$, and a vector space of dimension
$m$ over $\FF_q$, namely, a projective space denoted $PG(m-1,q)$.
Note that we will adopt a row vector convention for the rest of the paper.

\begin{definition}\label{def:spread}
Let $\mathcal{P}$ be a projective space. A $t$-{\em spread} of
$\mathcal{P}$ is a set $\mathcal{S}$ of $t$-dimensional subspaces
of $\mathcal{P}$ which partitions $\mathcal{P}$. That is, every point of
$\mathcal{P}$ is contained in exactly one $t$-space of $\mathcal{S}$.
\end{definition}

If $\mathcal{P}$=PG($m-1,q$) is a finite projective space, then a
$t$-spread can only exist if the number of points of a $t$-space
divides the number of points of the whole space, i.e., if
$\frac{q^{t+1}-1}{q-1} |~ \frac{q^m -1}{q-1}$ and hence $(q^{t+1}-1)|
(q^m -1)$, which holds if and only if $(t+1)|~m$. Andr\'e \cite{Andre}
showed that this necessary condition is also sufficient.

\begin{theorem}\cite{Eisfeld}
In PG($m-1,q$), a $t$-spread exists if and only if $t+1|~m$.
\end{theorem}

A systematic construction of spreads can be obtained through field
extensions as follows.
Suppose that $t+1 |~ m$. Consider the finite fields $F_0 = \FF_q$,
$F_1 = \FF_{q^{t+1}}$ and $F_2 = \FF_{q^m}$. Then $F_0 \subseteq
F_1 \subseteq F_2$. The field $F_2$ is an $m$-dimensional vector space
$V$ over $F_0$. The subspaces of $V$ form the projective space
$\mathcal{P}$=PG($m,q$). The field $F_1$ is a $(t+1)$-dimensional
subspace of $V$ and hence a $t$-dimensional (projective) subspace of
$\mathcal{P}$. The same holds for all cosets $a F_1$, $(a \in F_2)$.
These cosets partition the multiplicative group of $F_2$. Hence they
form a $t$-spread of $\mathcal{P}$.

\begin{example}\label{ex:part} \rm
Take as base field $F_0 = \FF_{2} $, i.e., the alphabet is $\{0,1\}$.
In order to obtain planes, we consider $1$-spread, i.e., $t=1$ and hence
$F_1 = \FF_{4}$. Finally, assume $m=4$, that is $F_2 = \FF_{16}$:
\[
\begin{diagram}
\dgARROWLENGTH=1em
\node{F_2 = \FF_{16}} \arrow{s,r,-}{\frac{m}{2}} \\
\node{F_1 = \FF_{4}} \arrow{s,r,-}{2} \\
\node{F_0 = \FF_{2}}
\end{diagram}
\]
Denote by $\FF_{q}^*$ the multiplicative group of $\FF_{q}$. Recall that
$\FF_{q}^*$ is a cyclic group. Let $\omega$ and $\nu$ be the respective generators of $F_2^*$ and $F_1^*$.
We have that $\nu$ is
an element of order 3 contained in $F_2$, so $\nu = \omega^5$. Thus
$F_1^*$ can be written $ F_1^* =\{ 1, \omega^5, \omega^{10} \}$.
As $F_2^*$ can be written
\[
F_2^*   =  \{ \omega^i \}_{i=1}^{15} 
        =  \{ \omega^i, \omega^{5+i} , \omega^{10+i} \}_{i=1}^5 
       =   \coprod_{i=1}^5 \omega^i \FF_{4}^*,
\]
we have a partition of $\FF_{16}$ into cosets of the form
$\omega^i \FF_{4}^* $, $i=1,\ldots,5$.
These five cosets define five disjoint
planes. More precisely, $\FF_{16}$ can be decomposed into direct sums
of $\FF_{2}$:
\[
\FF_{16}= \FF_{4} \oplus \nu \FF_{4} = \FF_{2} \oplus \nu \FF_{2}
         \oplus \omega \FF_{2} \omega \nu \FF_{2},
\]
so that each element of $\FF_{16}$ can be written as a $4$-tuple. For
example, the coset $\omega \FF_{4}^*$ contains the elements $\omega,
\omega \nu, \omega \nu^2$. As $\nu^2 =\nu +1$, $\omega \nu^2$ is the sum
of the two other points. Thus writing $\omega = (0,0,1,0)$ and $\omega \nu = (0,0,0,1) $,
we finally get that the plane defined by the coset $\omega \FF_{4}^*$ is
$\{ (0010), (0001), (0011)\}$.
\end{example}

%
%

\section{Code Construction}

Recall that our goal is to encode an object of size $B$ to be stored over $n$ nodes,
each of storage capacity $\alpha$, such that each failure can be repaired by contacting any $d$ live nodes,
$d\geq 2$. We denote by $PSRC(n,k)$ the self-repairing code with parameters $n$ and $k$ obtained
from a spread construction. 

We will assume for simplicity that we work over the base field $\FF_2$, though spreads can be constructed over larger alphabets.

\subsection{Setting the Parameters and Encoding}

1) We first set $m=B$, so that we are working with elements in $F_2=\FF_{q^B}$, that is $B$-dimensional vectors over $\FF_2$.

2) Consider a $t$-spread $\Sc$ formed of $t$-dimensional subspaces of $\Pc$ such that $t+1|B$. In particular, take $F_1=\FF_{q^{t+1}}$. Since every subspace is a $(t+1)$-dimensional vector space over
$\FF_2$, it is described by a $\FF_2$-basis containing $(t+1)$ vectors. We thus set $t+1=\alpha$, and assign to each node an $\FF_2$-basis containing $\alpha$ vectors. The number of nodes that will store the object is consequently (at most)
\[
n =\frac{2^B-1}{2^\alpha-1}.
\]
Since we must take $\alpha|B$, that is $B=b\alpha$, we can further write
\begin{equation}\label{eq:n}
n=\frac{2^{b\alpha}-1}{2^{\alpha}-1}=1+2^\alpha+(2^\alpha)^2+\ldots+(2^\alpha)^{b-1}.
\end{equation}

3) Let us denote by $v_i$ the collection of all $n\alpha$ vectors, ordered such that $v_1,\ldots,v_\alpha$
correspond to the first node, $v_{\alpha+1},\ldots,v_{2\alpha}$ to the second node, etc.
What the $i$th node will store is actually
\[
\{B v_{i\alpha+1}^T,\ldots,B v_{(i+1)\alpha}\}
\]
for a total storage of $\alpha$.

\begin{example}\label{ex:setting}\rm
Consider the partition described in Example \ref{ex:part}, where we recall that
$\nu^4=\nu+1$, $|\FF_{16}^*|=15$, $\nu^{15}=1$
$\omega^2=\omega+1$, $|\FF_4^*|=3$, $\omega^3=1$, $\omega=\nu^5=\nu^2+\nu$.

The final partition of the space is thus:
\begin{eqnarray*}
\FF_4^* & = & \{ (1000),(0110),(1110) \} \\
\nu\FF_4^* & = & \{(0100),(0011),(0111)\} \\
\nu^2\FF_4^* & = & \{ (0010),(1101),(1111)\} \\
\nu^3\FF_4^* & = & \{ (0001),(1010),(1011)\} \\
\nu^4\FF_4^* & = & \{(1100), (0101),(1001) \}
\end{eqnarray*}

\begin{table}\label{tab:B6}
\begin{tabular}{c|c|c}
$N_1$ to $N_7$ & $N_8$ to $N_{14}$ & $N_{15}$ to $N_{21}$\\
\hline
(100000),(110111) & (011000),(001110)&(001010),(110100)\\
(010000),(101011) & (001100),(000111)&(000101),(011010)\\
(001000),(100101) & (000110),(110011)&(110010),(001101)\\
(000100),(100010) & (000011),(101001)&(011001),(110110)\\
(000010),(010001) & (110001),(100100)&(111100),(011011)\\
(000001),(111000) & (101000),(010010)&(011110),(111101)\\
(110000),(011100) & (010100),(001001)&(001111),(101110)\\
\end{tabular}
\caption{Basis vectors for the scenario where we have $B=6$,
$\alpha=2$, $n=1+2^2+(2^2)^2=21$ nodes $N_1,\ldots,N_{21}$.}
\end{table}

This corresponds to the code parameters $B=4,~\alpha=2,~n = 1+2^2=5$ from (\ref{eq:n}).
Let us denote by $N_i$, $i=1,\ldots,5$ the 5 storing nodes, with storage capacity $\alpha=2$,
and by $\ov=(o_1,o_2,o_3,o_4)$ the object to be stored.
For example, we can use the basis vectors as follows:
\[
\begin{array}{ccc}
\mbox{node} & \mbox{basis vectors} & \mbox{data stored} \\
\hline
N_1 & v_1=(1000),~v_2=(0110)   & \{ o_1,o_2+o_3\}  \\
N_2 & v_3=(0100),~v_4=(0011)   & \{ o_2, o_3+o_4 \}  \\
N_3 & v_5=(0010),~v_6=(1101)   & \{ o_3, o_1+o_2+o_4 \}  \\
N_4 & v_7=(0001),~v_8=(1010)   & \{o_4, o_1+o_3 \}  \\
N_5 & v_9=(1100),~v_{10}=(0101)& \{ o_1+o_2, o_2+o_4 \} \\
\end{array}
\]
\end{example}

Furthermore, the first available parameters are summarized in Table \ref{tab:param}.

\begin{table}\label{tab:param}
\begin{center}
\begin{tabular}{c|c|c}
$B=b\alpha$ & $\alpha$  & $n=1+2^\alpha+\ldots (2^\alpha)^{b-1}$ \\
\hline
4 & 2 & 5 \\
6 & 2 & 21 \\
6 & 3 & 9 \\
8 & 2 & 85\\
8 & 4 & 17\\
\end{tabular}
\caption{
Set of some small available parameters for $PSRC(n,k)$.}
\end{center}
\end{table}

\subsection{Repair}

We now need to make sure that the above coding strategy allows for object retrieval and repair.
We start with repair of data stored in one storage node. It was shown in \cite{SRC} for HSRC that it is possible to repair data for one node by contacting $d=2$ nodes, and there are $(n-1)/2$ such choices of 2 nodes that allow repair. This holds also for PSRC.

\begin{lemma}
Suppose we have $n$ nodes, each storing $\alpha$ pieces of data encoding an object using $PSRC(n,k)$. Then if one node $N_l$ fails, it is possible to repair it by contacting $d=2$ nodes. More precisely, for any choice of node $N_i$ among the remaining $n-1$ live nodes, there exists at least one node $N_j$ such that $N_l$ can be repaired by downloading the data stored at nodes $N_i$ and $N_j$.
\end{lemma}
\begin{proof}
The $l$th node $N_l$ stores a subspace of the form $\nu^l\FF_{2^\alpha}^*$, $l=1,\ldots,n$.
Let us assume this $l$th node fails, and a new comer joins. It contacts any node, say $N_i$.
Since $N_i$ stores $\nu^i\FF_{2^\alpha}^*$, we need to show that there exists a node $N_j$ such that
\[
\nu^i\FF_{2^\alpha}^*\coprod \nu^j\FF_{2^\alpha}^*
\]
repairs $N_l$. Now
\[
(\nu^i+\nu^l)\FF_{2^\alpha}^*\subset\nu^i\FF_{2^\alpha}^*\coprod \nu^l\FF_{2^\alpha}^*
\]
so we can take $j$ such that $\nu^j=\nu^i+\nu^l$. By combining the data stored at node $N_i$ and $N_j$, we
thus get
\[
\nu^i\FF_{2^\alpha}^*\coprod (\nu^i+\nu^l)\FF_{2^\alpha}^*
\]
which contains $\nu^l\FF_{2^\alpha}^*$.
\end{proof}

\begin{example}\rm
Let us continue with Example \ref{ex:setting}.
If say $N_1$ fails, the data pieces $o_1$ (corresponding to the basis vector $(1000)$) and $o_2+o_2$ (corresponding to
the basis vector $(0110)$) are lost. A new node joining the network can contact nodes $N_3$ and $N_4$, from which
it gets respectively $v_5=(0010)$, $v_6=(1101)$ and $v_7=(0001)$, $v_8=(1010)$.
Now $v_8+v_5$ gives $(1000)$ while $v_8+(v_6+v_7)$ gives $(0110)$.
\end{example}

Actually, in general, the redundancy for self-repair provided by $PSRC$ is even stronger than that of $HSRC$, as we now illustrate.

\begin{lemma}\label{lem:B6}
Suppose we have $n=21$ nodes, each storing $\alpha=2$ pieces of data, encoding an object of size $B=6$ using  $PSRC(21,3)$, as summarized in Table \ref{tab:B6}. Then if one node $N_l$ fails, for any choice of node $N_i$ among the remaining $20$ live nodes, there exists three nodes $N_{j_1}$, $N_{j_2}$, $N_{j_3}$ such that $N_l$ can be repaired by downloading the data stored at either nodes $N_i$ and $N_{j_1}$, or $N_i$ and $N_{j_2}$, or even $N_i$ and $N_{j_3}$.
\end{lemma}
\begin{proof}
Recall that $\omega$ is the generator of the cyclic group $\FF_4^*$.
We have that node $N_l$ stores $\nu^l\FF_4^*$, and $N_i$ similarly stores $\nu^i\FF_4^*$.
Now
\[
\begin{array}{c}
\nu^l\FF_4^*\coprod \nu^i\FF_4^* = \\
\{
\nu^i+\nu^l, \nu^i\omega+\nu^l\omega,\nu^i+\nu^i\omega+\nu^j+\nu^j\omega \\
\nu^i,\nu^i\omega,\nu^i+\nu^i\omega, \\
\nu^l,\nu^l\omega,\nu^l+\nu^l\omega, \\
\nu^i+\nu^l\omega, \nu^i+\nu^i\omega+\nu^l,\nu^i\omega+\nu^l+\nu^l\omega \\
\nu^i\omega+\nu^l,\nu^i+\nu^i\omega+\nu^l\omega,\nu^i+\nu^l+\nu^l\omega
\}=\\
(\nu^i+\nu^l)\FF_4^*\coprod \nu^i\FF_4^*\coprod\nu^l\FF_4^*\coprod (\nu^i+\nu^l\omega)\FF_4^*
\coprod (\nu^l+\nu^i\omega)\FF_4^*.
\end{array}
\]
Take $j_1,j_2,j_3$ such that
\[
\nu^{j_1}=\nu^i+\nu^l,~\nu^{j_2}=\nu^i+\nu^l\omega,~\nu^{j_3}=\nu^l+\nu^i\omega.
\]
We have then
\begin{eqnarray*}
(N_i,N_{j_1}) & \Rightarrow & \nu^i\FF_4^* \coprod (\nu^i+\nu^l)\FF_4^*\supset\nu^l\FF_4^*,\\
(N_i,N_{j_2}) & \Rightarrow & \nu^i\FF_4^* \coprod (\nu^i+\nu^l\omega)\FF_4^*\supset\nu^l\FF_4^*,\\
(N_i,N_{j_3}) & \Rightarrow & \nu^i\FF_4^* \coprod (\nu^l+\nu^i\omega)\FF_4^*\supset \nu^l\FF_4^*.
\end{eqnarray*}
\end{proof}

This proof actually gives an algorithm to find the different pairs that repair a given failed node.

\begin{example}
Consider the code described in Table \ref{tab:B6}, and suppose that the node $N_1$ fails, and a new comer contacts node $N_4$ which stores $\nu^3\FF_4^*$. We have
\begin{eqnarray*}
\nu^l+\nu^i       &=1+\nu^3 = \nu^{21}\nu^{11}& \Rightarrow N_{12} \\
\nu^l\omega+\nu^i &=\omega+\nu^3= \nu^{21}\nu^9 & \Rightarrow N_{10} \\
\nu^l+\nu^i\omega &=1+\nu^3\omega= \nu^4 & \Rightarrow N_5.
\end{eqnarray*}
Thus the node $N_1$ can be repaired by contacting the following three pairs all involving $N_4$:
\[
(N_4,N_{12}),~(N_4,N_{10}),~(N_4,N_5).
\]
\end{example}

\begin{figure*}[htbp]
    \begin{center}
    \subfigure[\label{fig:oneminusrho}$1-\rho_x$ (determined using exhaustive enumeration)]{\includegraphics[scale=0.6]{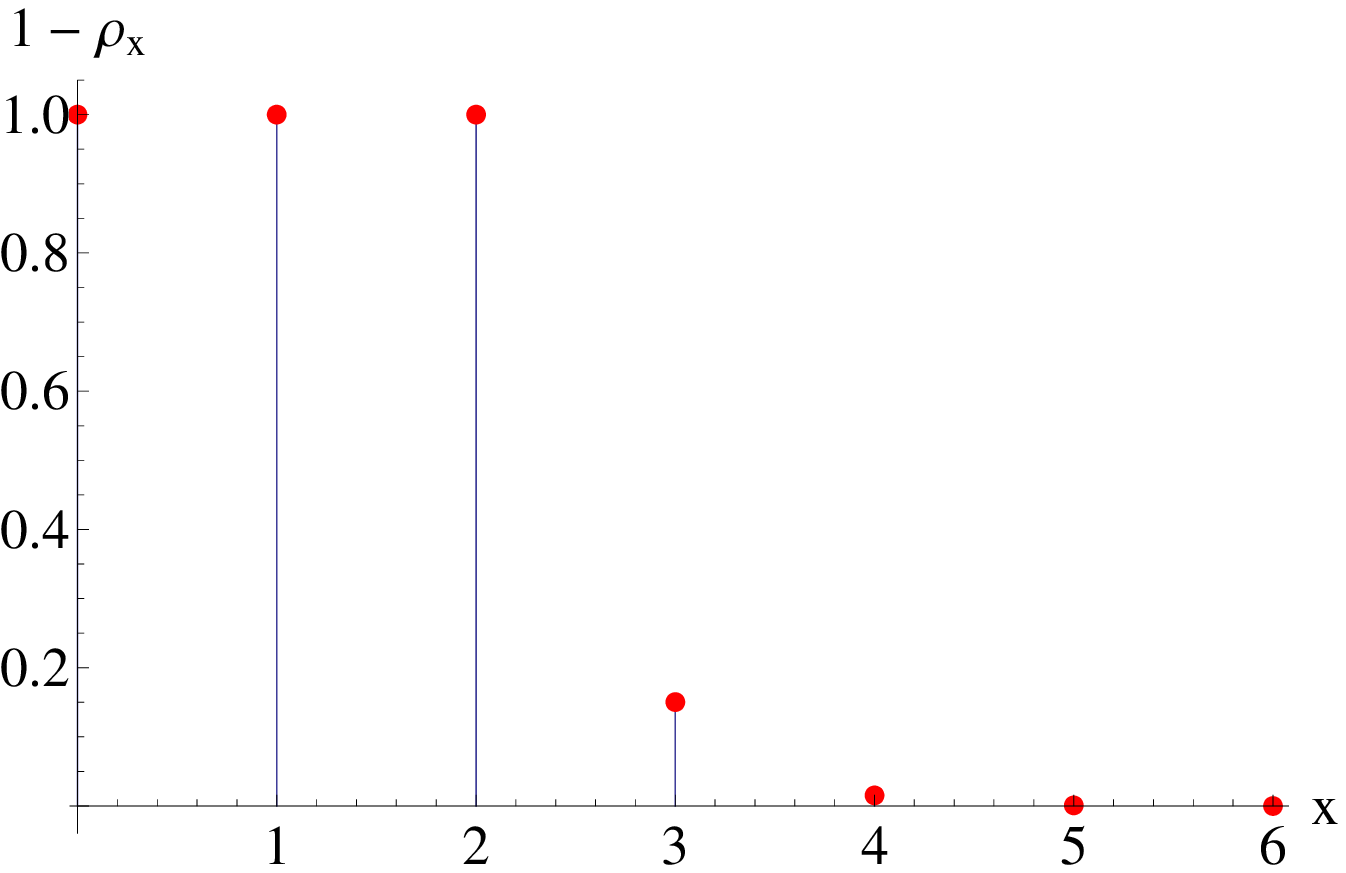}}\hspace{5mm}
    \subfigure[\label{fig:staticresilience}Static resilience (determined numerically)]{\includegraphics[scale=0.6]{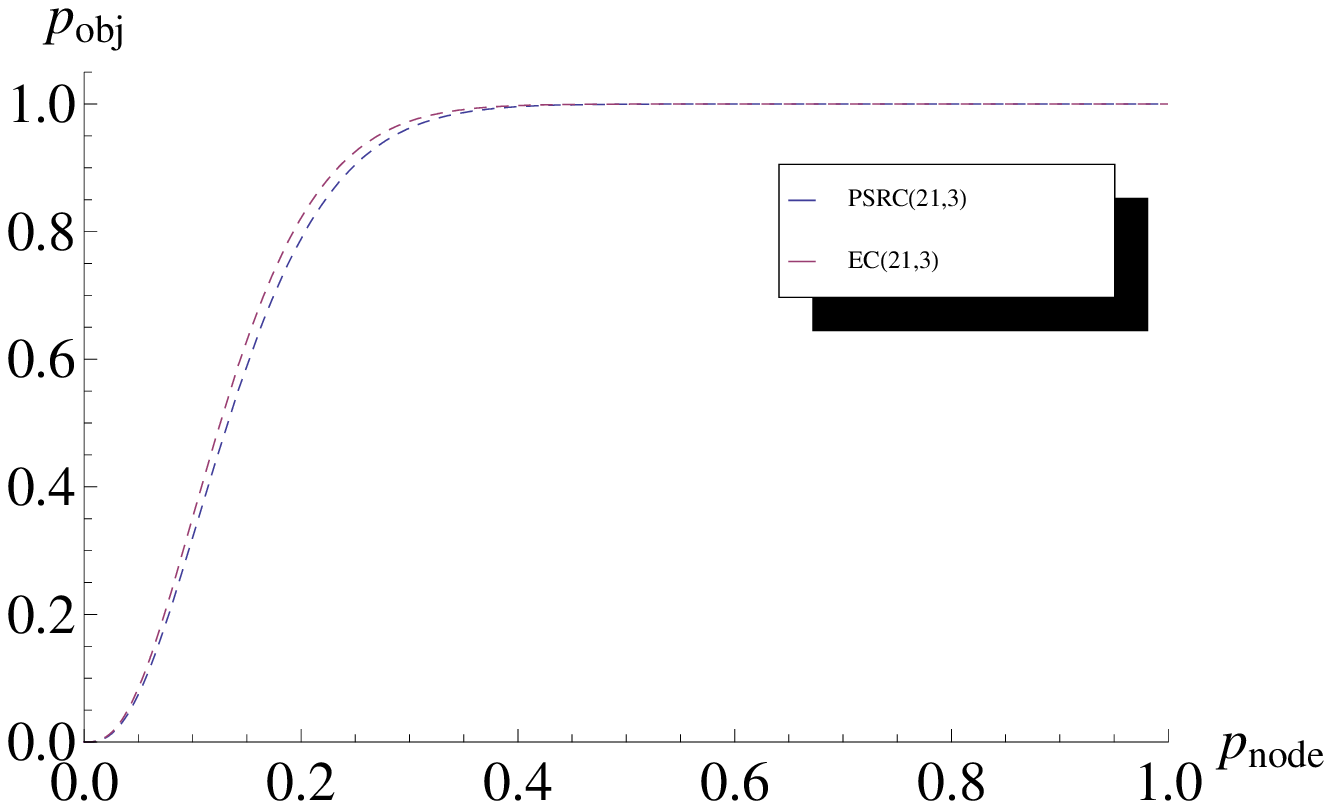}}
  \end{center}
  \caption{Results for $PSRC(21,3)$}
  \label{fig:results}
\end{figure*}

\subsection{Object Retrieval}

If a data collector connects to any choice of $k$ nodes, then he can access upto $k\alpha$ blocks, while trying to reconstruct an object of size $B$. Thus, $k \geq B/\alpha$. Note that in the examples considered in this paper, $k = B/\alpha$. 

\begin{lemma}\label{lem:k2}
If $k=2$, then the object can be retrieved from any choice of $k=2$ nodes,
in which case, we may see $PSRC(n,k)$ as a MDS code.
\end{lemma}
\begin{proof}
If $k=2$, then each node stores $\alpha=B/2$ linearly independent vectors. Pick any two nodes say $N$ (containing $v_1,\ldots,v_\alpha$) and $N'$ (similarly storing $u_1,\ldots,u_\alpha)$. Suppose that there exists a vector $v$ in $N$ which is linearly dependent of some vectors in $N'$:
\[
v=\sum_{i=1}^\alpha a_iv_i+\sum_{j=1}^\alpha b_ju_j.
\]
Since $v\in N$ and $\sum_{i=1}^n\alpha a_iv_i \in N$, it must be that
$\sum_{j=1}^\alpha b_ju_j\in N$, a contradiction since $N$ and $N'$ are non-intersecting by the definition of spread.
\end{proof}

To recover the object, the data collector just solves the system of linear equations in $\ov$.

%
%

In general, when $k\geq 3$, SRC codes are not maximum distance separable (MDS). A static resilience analysis provides an estimate of how much deterioration the system may suffer due to the lack of the maximum distance separability.

\emph{Static resilience} of a distributed storage system is defined as the probability that an object, once stored in the system, will continue to stay available without any further maintenance, even when a certain fraction of individual member nodes of the distributed system become unavailable. Let $\probup$ be the probability that any specific node is available. Then, under the assumptions that node availability is $i.i.d$, and no two fragments of the same object are placed on any same node, we can consider that the availability of any fragment is also $i.i.d$ with probability
$\probup$. The probability $\objup$ of recovering the object is then
\[
\objup = \sum_{x=k}^n \rho_x C_{x}^{n} \probup^x (1-\probup)^{n-x},
\]
where $\rho_x$ is the conditional probability that the stored object can be retrieved by contacting an arbitrary $x$ out of the $n$ storage nodes.

For $(n,k)$ MDS erasure codes, $\rho_x$ is a deterministic and binary value equal to one for $x \geq k$, and zero for smaller $x$. For self-repairing codes, the value is probabilistic. In Fig. \ref{fig:oneminusrho} we show for our toy example $PSRC(21,3)$ the probability that the object cannot be retrieved, i.e., $1-\rho_x$, where the values of $\rho_x$ for $x \geq k$ were determined by exhaustive search.\footnote{$\rho_x$ is zero for $x<k$ for PSRC also.}

In particular, one can list 17 unique groups of 5 nodes, whose all together 10 basis vectors generate a matrix with rank less than 6, out of the ${21\choose 5}= 20349$ unique groups of 5. This means that if we choose any 5 arbitrary nodes, the object still cannot be retrieved with a probability of 0.00083, which is rather negligible. Similarly, if we chose any arbitrary 3 nodes, the probability of unretrievability is 0.150375. In contrast, for MDS codes, the object will be retrievable from the data available at any arbitrary three nodes. Of-course, this rather marginal sacrifice provides PSRC an incredible amount of self-repairing capability. For any one node lost, as shown earlier in Lemma \ref{lem:B6}, one can choose any of the twenty remaining live nodes, and pair it with three other nodes, and regenerate the lost data.



In Fig. \ref{fig:staticresilience} we compare the static resilience $\objup$ for $PSRC(21,3)$ with respect to what could be achieved using a MDS $EC(21,3)$. The values were determined numerically, using the $\rho_x$ values evaluated as mentioned above. We note that in practice a MDS erasure code may or not exist with the specific $(n,k)$ parameters. More importantly, we notice that the degradation of static resilience of $PSRC(21,3)$ to achieve the self-repairing property is marginal with respect to that of a MDS erasure code, if such a code were to/does exist.

%
%
%

\section{Further discussions}

We point out a few more properties of the proposed codes.

{\em Systematic Like Code:} It is usually appreciated from an implementation perspective to use a systematic code, since it makes the object retrieval immediate. We notice that though our code is not systematic, we can however contact $B$ specific nodes (instead of $k$), namely those storing as pieces each of the canonical basis vectors of $\FF_{2^B}$ to reconstruct the object in a systematic manner.

{\em Bandwidth cost for regeneration:} Unlike HSRC, the PSRC encoded blocks are not atomic, and instead comprise of $\alpha$ pieces. Thus, similar to regenerating codes, one could also expect to regenerate an encoded block piece-by-piece, by contacting more (larger $d$) number of nodes. For example, when using $PSRC(21,3)$, if the data for node $N_1$ needs to be regenerated, one could do so by contacting two nodes and downloading four pieces (units) of data, as we have already seen. One could instead also contact $d=3$ nodes, and regenerate the two lost pieces by downloading only three units of data. For instance, by downloading (010000) from $N_2$, (110000) from $N_7$ and (000111) from $N_9$.

As noted previously, for our examples, $\alpha=B/k$, corresponding to what is known as the Minimum Storage Regeneration (MSR) point for regenerating codes. At MSR point, a node needs to contact $d \geq k$ nodes, and download $\frac{B}{k(d-k+1)}$ data from each, resulting in a total download of $\frac{Bd}{k(d-k+1)}$ data. Thus, for the same choices of $\alpha,B,k$ and with $d=3$, one would need to download 6 units of data, and for $d=4$, one would need to download 4 units of data, while $d=2$ is not allowed. Thus, for the regeneration of one lost node, PSRC can outperform regenerating codes both in terms of absolute bandwidth needed, as well as the number of nodes needed to carry out such regeneration, moreover, for upto $(n-1)/2$ failures, the regeneration overhead per node's data stays constant for PSRC. It of-course needs to be noted that, in order to achieve these very interesting performance, we sacrificed the MDS property. In practice, this sacrifice however has marginal impact, as can be observed from the resulting codes' static resilience.



%
%
\section{Concluding remarks}

In this work, we showed the existence of another instance of self-repairing codes, which are codes
tailor made to meet the peculiarities of distributed networked storage. The proposed code family in this paper is based on constructions of spreads from projective geometry. We provided a preliminary study of the properties of this new family, demonstrating that they outperform existing code families both in several quantitative as well as qualitative metrics. Further analysis to comprehend and harness these codes in practical settings are currently under investigation.

%
%
\section*{Acknowledgement}
F. Oggier's research for this work has been supported by the Singapore National Research Foundation grant NRF-CRP2-2007-03. A. Datta's research for this work has been supported by AcRF Tier-1 grant number RG 29/09.

\end{document}